\documentclass[10pt,journal,compsoc]{IEEEtran}
\ifCLASSOPTIONcompsoc
  \usepackage[nocompress]{cite}
\else
  \usepackage{cite}
\fi
\ifCLASSINFOpdf
\else
\fi
\usepackage{flushend}
\usepackage{amsmath}
\usepackage{amsthm}
\usepackage{amsfonts}
\usepackage{amssymb}
\usepackage{lipsum}                     % Dummytext
\usepackage{xargs}   

\usepackage[pdftex,dvipsnames]{xcolor}  % Coloured text etc.
\usepackage{algorithm}
\usepackage{algorithmic}
\usepackage{graphicx}
\usepackage{caption}
\usepackage{color}

\newtheorem{theorem}{Theorem}

\newtheorem{definition}{Definition}
\newtheorem{example}{Example}
\newtheorem{assumption}{Assumption}

\usepackage{subcaption}

\providecommand{\U}[1]{\protect\rule{.1in}{.1in}}
\usepackage[colorinlistoftodos,prependcaption,textsize=tiny]{todonotes}
\newcommandx{\comment}[2][1=]{\todo[linecolor=red,backgroundcolor=red!25,bordercolor=red,#1]{#2}}

\begin{document}

\title{Broadcast Distributed Voting Algorithm in Population Protocols}

\author{Hamidreza~Bandealinaeini and Saber~Salehkaleybar,~\IEEEmembership{Member,~IEEE}
\thanks{Both authors are with Learning and Intelligent Systems Laboratory, Department
	of Electrical Engineering, Sharif University of Technology, Tehran, Iran.
	E-mails: hr\_bandeali@ee.sharif.edu saleh@sharif.edu
	(corresponding author). Webpage: http://lis.ee.sharif.edu.}}% <-this % stops an unwanted space
\IEEEtitleabstractindextext{%
\begin{abstract}
We consider the problem of multi-choice majority voting in a network of $n$ agents where each agent initially selects a choice from a set of $K$ possible choices. The agents try to infer the choice in majority merely by performing local interactions. Population protocols provide a framework for designing pairwise interactions between agents in order to perform tasks in a coordinated manner. In this paper, we propose ``Broadcasting Population Protocol" model as a counterpart model of conventional population protocols for the networks that each agent can send a message to all its neighbors simultaneously. We design two distributed algorithms for solving the multi-choice majority voting problem in the model of broadcasting population protocols. We prove the correctness of these algorithms and analyze their performance in terms of time and message complexities. Experiments show that the proposed algorithm improves both time and message complexities significantly with respect to previous algorithms proposed in conventional population protocols and they can be utilized in networks where messages can be transmitted to a subset of agents simultaneously such as wireless networks. 
\end{abstract}

}

\maketitle

\IEEEdisplaynontitleabstractindextext

% \fi
\IEEEpeerreviewmaketitle

\IEEEraisesectionheading{\section{Introduction}\label{sec:introduction}}

%%\IEEEPARstart{T}{his} demo file is intended to serve as a ``starter file'' for IEEE Computer Society journal papers produced under \LaTeX\ using IEEEtran.cls version 1.8b and later.
%%I wish you the best of success.
\IEEEPARstart{C}{ommunication} protocols can be categorized into two main groups \cite{lynch1996distributed}: shared-memory and message-passing. In shared-memory protocols, there exist memories which are shared among agents and they can write on these memories and read from to communicate with each other. In message-passing protocols, agents in a network communicate with each other by transmitting messages through channels. In multi-agent systems with limited resources, message-passing protocols can be further divided into three main models \cite{navlakha2015distributed}: Beeping, Stone-age, and Population Protocols. In beeping models, agents transmit messages by sending beep signals in synchronous time slots. Thus, in each time slot, agents can send beep signals or be silent. In Stone-age model, agents are represented by finite state machines and they use rich size messages for communication but their memories are restricted to a limited capacity. In population protocols, agents are modeled as finite state machines with limited resources that interact in a pairwise manner to update their states. The goal is to compute a function globally in the network by performing pairwise interactions between agents. Recently, population protocols are considered as a mathematical model to discover the computational power of biological information processing systems such as chemical reaction networks \cite{chen2014deterministic} and gene regulatory networks \cite{navlakha2015distributed}.
In the literature of population protocols, several distributed algorithms have been proposed in the model of population protocols for various problems like leader election \cite{alistarh2015polylogarithmic,doty2018stable,angluin2005self, berenbrink2018simple}, network size approximation and counting \cite{aspnes2016time,beauquier2015space,izumi2014space,beauquier2007self}, and majority voting \cite{angluin2008simple,benezit2011distributed,salehkaleybar2015distributed,gasieniec2017deterministic, alistarh2018space, alistarh2015fast}.

%\subsection{Related Work}
%Population protocol is a kind of weak communication model and it can prevent waste of energy or dropping of messages which are the results of transferring messages with large sizes in communications.
%Consequently, these protocols have being used recently and various algorithms have been proposed for solving distributed problems like leader election [][], distributed voting [], maximal independent set, minimum connected dominating set [], network size approximation and counting [], [], [], [], deterministic rendezvous problem [], naming  [], membership problem [], [], broadcasting [], [], [], and consensus. 

In the problem of distributed majority voting which is the main focus of this paper, agents have initial choices, and they interact with each other until their states converge to a choice which is in majority among initial choices. Due to limited resources of agents in population protocols, most of previous works have focused on designing automaton with small memory. In \cite{angluin2008simple}, Angluin et al. proposed a 3-state automaton for the problem of binary majority voting and showed that agents can approximately find the majority vote in the sense that with high probability, they reach consensus on the majority vote if the difference between number of supporters of the two choices is large enough. Later, B\'en\'ezit et al.\cite{benezit2009interval} proposed an automaton with four states that can solve the problem of binary voting exactly. They also extended the idea in designing 4-state automaton to ternary and quaternary voting with 15 and 100 states, respectively \cite{benezit2011distributed}. For more than four choices, they also suggested to execute binary voting between any pair of choices in parallel which requires $O(2^{K(K-1)})$ number of states per agent where $K$ is the number of choices. Salehkaleybar et al. \cite{salehkaleybar2015distributed} proposed Distributed Multi-Choice Voting/Ranking (DMVR) algorithm which reduces the number of states to $O(K\times 2^{K-1})$ for the problem of majority voting. In the proposed algorithm, the interaction between pair of agents is simply based on intersection and union operations. Moreover, the proposed algorithm is state-optimal for the problem of ranking. Later, Alistarh et al.\cite{alistarh2017time} showed a trade-off between the number of states per agent and convergence time for the binary majority voting problem. In particular, the problem can be solved in time $O(\log(n)/s+\log(n)\log(s))$ where $n$ and $s$ are the number of agents and number of states per agent, respectively. 
Recently, Gasieniec et al.\cite{gasieniec2017deterministic} designed an automaton with $O(K^6)$ number of states to solve majority voting by running multiple binary voting algorithms in parallel.
%\subsection{Our Contribution}
%TODO: I've edited Broadcast BMVR with B-DMV 

In this paper, we first introduce ``Broadcasting Population Protocol" model in which each agent can interact with all its neighbors simultaneously instead of only one of them as it is in conventional population protocol model. This model can be considered in networks where each agent can send its massage to all its neighbor simultaneously such as transmitting signals in a wireless medium. Due to simple structure of interactions in the DMVR algorithm \cite{salehkaleybar2015distributed}, we can adopt this algorithm for the broadcast setting which call it ``Broadcast Distributed Multi-Choice Voting  (B-DMV)" algorithm. Experiments show that the broadcast version of DMVR improves average running time significantly with respect to pairwise version. However, the number of transmitted messages in B-DMV is still close to the pairwise one. Thus, we propose two variants of B-DMV algorithm which improve both time and message complexities dramatically. We prove the correctness of these algorithms and analyze their performance in terms of time and message complexities. We also provide experiment results of the proposed algorithms for different network topologies and network sizes.
%TODO : I've edited the number of messages in pairwise and B_DMV (Please check it)

The rest of this paper is organized as follows. In Section 2, we present the model of broadcasting population protocols and describe B-DMV algorithm. We also provide the correctness of this algorithm. In Section 3, we propose two variants of B-DMV and show the correctness of these algorithms. In Section 4, we analyze time and message complexities of the proposed algorithms. We conduct experiments to evaluate
the performance of the proposed algorithms in
Section 5 and conclude in Section 6.
 

\section{Model}
\subsection{Population Protocol}

Consider a network of agents whose programs are identical and can be represented as a finite state machine. Let $G=(V,E)$ be the graphical representation of the network where $V=\{1,\cdots,n\}$ is the set of vertices (corresponding to the agents) and $E\subseteq V\times V$ is the set of edges in which $(i,j)\in E$ if agent $i$ can directly interact with agent $j$. We denote the set of neighbors of agent $i$ by $N(i)$. The goal is to carry out a computation like leader election, consensus, or majority voting by merely local interactions.

In population protocol, agents initially get their input values from a specific set $\mathcal{X} $. Then by applying input function like $\mathcal{I:X\rightarrow Q}$, their initial states are obtained from their inputs where $\mathcal{Q}$ is the set of states in each agent. 
%A vector whose its elements are state of agents is%
We represent the states of agents in a vector $C=[q_1,\cdots,q_n]$ and call it ``Configuration Vector" where $q_i$ is the state of agent $i$. We assume that there exists a scheduler which picks  two neighbor agents at each time $t$ to interact with each other. 
%In fact, at each time, pairs of agents which are selected by a scheduler interact and update their states.%
Suppose that agent $i$ with state $q_i$ and agent $j$ with state $q_j$ are selected to have an interaction. These two agents update their states to $q'_i, q'_j$ respectively according to a transition rule like $\delta : Q \times Q \longrightarrow Q \times Q$. Thus $i$-th and $j$-th entries in the configuration vector are changed from $q_i$ and $q_j$ to $q'_i$ and $q'_j$, respectively.
% In other words if we suppose that the previous configuration of protocol which had $Q_i$ and $Q_j$ was $C$, after this interaction the protocol goes from configuration $C$ to $C'$ which have $Q_i'$ and $Q_j'$ instead of $Q_i$ and $Q_j$. 
The execution of population protocol is an infinite sequence of configurations $C_0, C_1, ... $ where $C_0$ is the initial configuration vector and each $C_{i+1}$ is obtained from $C_i$ by just one pairwise interaction. During the execution, at each time $t$, by applying the output function $\mathcal{O:Q\rightarrow Y}$ to the states of the agents, the outputs of agents can be obtained where $\mathcal{Y}$ is the set of outputs.

In population protocols, it is commonly assumed that the scheduler is fair. More specifically, for any configuration $C'$ that can be obtained from a configuration $C$ by a pairwise interaction, if configuration $C$ is infinitely often appears in the execution, then the configuration $C'$ should also appear infinite time in execution.
% The above paragraph explains basic model of population protocols which uses pairwise interaction. An important and essential point which should be mentioned is about the property of scheduler which choose pair of agents for interaction. The scheduler needs to be fair. For definition of fairness suppose configuration $C'$ which is obtainable from configuration $C$. An execution is fair if and only if configuration $C$ is infinitely in the execution, configuration $C'$ should be infinitely in the execution too. %
Please note that this property is necessary for convergence of network. Otherwise, the scheduler can separate the network into two or more disjoint components such that there is no interaction between components. Herein, to have rigorous analysis of time and message complexities, it is assumed that each agent has a local clock which ticks according to a Poisson process with rate one. When a local clock of agent $i$ ticks at time $t$, it wakes up and selects one of its neighbor $j\in {N(i)}$ randomly to interact with it. By this assumption, the fairness property of the scheduler is guaranteed. 

In the next part, we introduce a new kind of population protocol model where agent interacts with all of its neighbors. In order to distinct between the new model with conventional population protocol, we call the former and latter ones \textbf{\textit{``Pairwise PP"}} and \textbf{\textit{``Broadcasting PP"}}, respectively.

%\subsection{Two type of interactions}
%In population protocols, instead of sending and receiving message, one agent interact with some other agents that are close to it and they update their states. By continuing this, the network get to the convergence. Generally, there are two type of interactions for nodes in the network; pairwise and broadcast. 
%\subsubsection{Pairwise (Population Protocol)}
%In pairwise interaction, at each time one node wakes up and it chooses one of its neighbours randomly to have an interaction with. As it has been mentioned, in population protocols nodes use pairwise interaction to update their states. For example if a node $i$ with the state $Qi$ interacts with a node $j$ with the state $Q_j$, here the interaction will be of kind of pairwise interaction and it is shown by $Q_i$ $\times$ $Q_j$ $\longrightarrow$ $\acute{Q_i}$ $\times$ $\acute{Q_j}$. It says that after this action, the new states of nodes become $\acute{Q_i}$ and $\acute{Q_j}$. 

\subsection{Broadcasting Population Protocol}

In broadcasting PP, similar to pair PP, each agent applies input function in order to obtain its initial state. There exits a scheduler that picks one of the agents at each time to interact with all of its neighbors simultaneously. All agents involving in the interaction update their states according to a transition rule $\delta_B^d : \underbrace{Q\times \cdots \times Q}_\text{d+1 times} \longrightarrow  \underbrace{Q\times \cdots \times Q}_\text{d+1 times}$ where $d=|N(i)|$. Please note that for each value of $d$, the specific transition rule $\delta_B^d$ is considered. The configuration vector and an execution of broadcasting PP can be defined to the ones in pairwise PP similarly.

%In the process of broadcast interaction, firstly the starter agent broadcasts a message to its neighbours, secondly its neighbours send their states to the starter one, and finally the starter determine their new states and broadcast these states to the neighbours randomly. These three steps are considered as one $operation$. In the following we use from this definition of operation to analyse time complexity.%

%\subsection{The model}
%We assume that our model is a network with $n$ agents which is a connected undirected graph. It is shown by $G=(V,E)$. The $V$ is a set which is used for the vertexes of graph and the set is $\mathcal{}V = \{ 1,...,n\}$. The $E$ is also a set which is used for the edges of graph. In fact, $\mathcal{}E\subset V\times V$ and $\mathcal{}(i,j)\in E$ if and only if node $i$ and node $j$ are connected directly. This is the topology of the network. Additionally, it is assumed that each node has a local clock that ticks according to a Poisson process with rate one. There are K different choices which initially each node can choose one of them form the choices set $\mathcal{}C = \{ c_1,...,c_k\}$. Let $\# c_k$ be the number of nodes that select the choice $c_k$ and $\rho \triangleq \frac{\#c_k}{n}$. The nodes have a memory which can save one of this choices. 

\subsection{Problem Definition}

Suppose that there are $K$ different choices $\mathcal{}C = \{ c_1,...,c_K\}$ where each agent initially chooses one of them as its input. Let $\# c_k$ be the number of agents that select the choice $c_k$ and $\rho_k \triangleq \frac{\#c_k}{n}$. The goal in majority voting problem is to find the choice with the most number of supporters; i.e, finding the index $k$ such that argmax$_k$ $\rho_k $.

In the following, we first describe DMVR algorithm \cite{salehkaleybar2015distributed} which can be utilized for solving majority problem in population protocol. Then we adopt this algorithm for broadcasting PP model and call it B-DMV algorithm.

\subsection{Description of DMVR Algorithm}

Assume that each agent $i$ has a value set $\mathcal{}v_i$ and a memory set $\mathcal{}m_i$ which are denoted by $\mathcal{}v_i(t)$ and $\mathcal{}m_i(t)$ at time $t$, respectively. At time $t=0$, both of these sets for each agent $i$ are set to the choice that agent has chosen. For example if agent $i$ chooses $c_1$, then $v_i(0)=\{c_1\}$ and $m_i(0)=\{c_1\}$ (see Algorithm \ref{Alg:1}). In each interaction, the value sets of agents $i$ and $j$ participating in the interaction are updated to the union and intersection of current value sets in these agents (see lines 4-7). Moreover, their memory sets are also set to their updated value sets if the size of corresponding value set in each agent is equal to one. 

\begin{algorithm}
	\caption{The Distributed Multi-choice Voting/Ranking Algorithm}
	\begin{algorithmic}[1] 
		\STATE Initialization : $v_i(0)=$ Initial vote of agent $i$, $m_i(0)=$ Initial vote of agent $i$\;
		\IF{agent $i$'s clock ticks at time $t$}
		\STATE agent $i$ contacts with a random neighbor, say agent $j$
		\IF{$|v_i(t)| \leq |v_j(t)|$}
		\STATE $v_i(t^+) := v_i(t)\cup v_j(t),  v_j(t^+) := v_i(t) \cap v_j(t).$
		\ELSE
		\STATE $v_i(t^+) := v_i(t)\cap v_j(t),  v_j(t^+) := v_i(t) \cup v_j(t).$
		\ENDIF
		\IF{$|v_i(t^+)|==1$}
		\STATE $m_i(t^+)=v_i(t^+)$
		\ENDIF
		\IF{$|v_j(t^+)|==1$}
		\STATE $m_j(t^+)=v_j(t^+)$
		\ENDIF
		\ENDIF
	\end{algorithmic}
\label{Alg:1}
\end{algorithm}

\begin{example}
	In this example, we describe input function $\mathcal{I} $, output function $\mathcal{O} $, and transition function $\mathcal{\delta} $ for the population protocol of DMVR algorithm which solves binary majority voting. 
	Consider the set of choices $\mathcal{X}=\{ c_1, c_2 \}$. The set of states is $\mathcal{Q} = \{ 0, 1, 1^+, 1^- \}$ in which state $0$ corresponds to value set $\{0\}$ and memory set $\{0\}$, state $1$ corresponds to value set $\{1\}$ and memory set $\{1\}$, state $0^+$ corresponds to value set $\{0, 1\}$ or $\emptyset$ and memory set $\{0\}$, and state $1^+$ corresponds to value set $\{0, 1\}$ or $\emptyset$ and memory set $\{1\}$. The input function $\mathcal{I}$ is applied to the input set as follows: $\mathcal{I}: c_1\rightarrow0 , c_2\rightarrow1$. The four rows from the transition rule table (out of 16 possibilities)  are: $(0,1) \rightarrow (0^+,1^+)$ , $(1,0) \rightarrow (1^+,0^+)$ , $(0,0) \rightarrow (0,0)$ , and $(1,1) \rightarrow (1,1)$. At each time the output function $O$ can apply to the states of agents and it maps $ 0,0^+$ to $c_1$ and $1,1^+$ to $c_2$. 
	%The output set is $\mathcal{Y}=\{ C_1, C_2 \}$ which here is like the input set. Population protocols usually do not halt and the outputs of agents converge to a stable configuration according to the inputs of agents.
\end{example}
%%To be familiar with the different functions and sets which we introduce in the previous section (like input function $\mathcal{I} $, output function $\mathcal{O} $, transition function $\mathcal{\delta} $, and ...), here there is an example of a network that tries to solve majority. Suppose a network whose its input set is $\mathcal{X}=\{ C_1, C_2 \}$. Then input function $I$ applies to input sets and results $I: C_1\rightarrow0 , C_2\rightarrow1$. Here the states set is $\mathcal{Q} = \{ 0, 1, 1^+, 1^- \}$ in which state $0$ means value set $\{0\}$ and memory set $\{0\}$, state $1$ means value set $\{1\}$ and memory set $\{1\}$, state $0^+$ means value set $\{0, 1\}$ and memory set $\{0\}$, and state $1^+$ means value set $\{0, 1\}$ and memory set $\{1\}$. The transition function $\delta$ here causes the interactions like $0,1 \rightarrow 0^+,1^+$ , $1,0 \rightarrow 1^+,0^+$ , $0,0 \rightarrow 0,0$ , and $1,1 \rightarrow 0,1$. At each time the output function $O$ can apply to the states of agents and it results $O: 0,0^+ \rightarrow C_1$ $1,1+ \rightarrow C_2$. The output set is $\mathcal{Y}=\{ C_1, C_2 \}$ which here is like the input set. Population protocols usually do not halt and the outputs of agents converge to a stable configuration according to the inputs of agents.

There exist two phases in the course of executing DMVR algorithm. In the first phase, the algorithm tries to reduce the number of value sets of size one by performing union and intersection operations between value sets of agents. More specifically, suppose that two agents $i$ and $j$ having value sets of size one, i.e., $|v_i(t)| = |v_j(t)| =1$, and $v_i(t) \neq v_j(t)$. If at time $t$ these two agents interact with each other, the number of value sets of size one in the network, will be decreased by two. This process continues until we reach to a configuration in which there is no value set of size one from the choices in minority and each value set is a subset of all value sets with larger size \cite{salehkaleybar2015distributed}. Furthermore, all the subsequent configurations have the above mentioned property. For example, consider two choices $c_1$ and $c_2$ and assume that the choice $c_1$ is in the majority. At the end of the first phase, the value sets of different agents can only be $\{\}$, $\{c_1\}$, and $\{c_1,c_2\}$. Due to the fact that $\{\}\subset\{c_1\}\subset\{c_1,c_2\}$, the union and intersection operation does not produce any new value set. In the second phase, the remaining value sets of size one for the majority will be spread through the network by updating memories of agents. Please note that updating value sets and memories are performed in parallel, but when all of the value sets of choices in minority are eliminated, the true answer will be only disseminated among agents' memories in the second phase.

\section{Description of B-DMV Algorithm}

In this section we first describe a naive solution called B-DMV algorithm in order to adopt DMVR algorithm to Broadcasting PP model. Experiments show that B-DMV improves average runtime while number of transmitted message does not improve considerably with respect to pairwise version. Therefore, we propose two modified versions of B-DMV algorithm and prove the correctness of these algorithms. Furthermore, we analyze their time complexities by mean-field theorem and show that the proposed algorithms improve the time and message complexities significantly.

\subsection{B-DMV Algorithm}

When an agent $i$ wakes up, it gathers the value sets of all its neighbors, $N(i)$. Then, it consolidates the value sets as follows: 
First, it counts the number of occurrence of each choice $c_k$ in all the received value sets which is denoted by $n_k$. 
Then it constructs a new temp value set containing choices with $n_k > 0$. Then it decreases all these $n_k's$ by one and send the temp value set to one of its neighbor agents that have not already received a new value set, randomly. It continues this procedure until all agents in $N(i)\cup \{i\}$ received their value sets. 
As it has been mentioned before, interactions in broadcasting PP are among one agent and all of its neighbors. 

The description of B-DMV algorithm is given in Algorithm \ref{alg:BDMV}.
%Consequently, a new algorithm needs to be introduced. Let $N(i)$ and $d_i$ be the set of neighbor agents and degree of agent $i$ respectively. 
%TODO : I have used randomly in line 4th above.
\begin{algorithm}[t]
	\caption{The B-DMV Algorithm}
	\begin{algorithmic}[1]
		\STATE Initialization : $v_i(0)=$ Initial vote of agent $i$, $m_i(0)=$ Initial vote of agent $i$\;
		\IF{Agent $i$'s clock ticks at time $t$}
		\STATE Agent $i$ contacts with all of its neighbors, i.e., $N(i)$.
		\STATE Consider $v_i$ as $v_{i_0}$ and the value sets of $N(i)$ as $v_{i_1},...,v_{i_d}$ 
		\STATE $n_k := | \{j|c_k \in v_j, j=i_0,...,i_d\} |$ \label{a2}
		\STATE {$S := N(i) \cup \{i\}$}
		\FOR {$u=1$ \TO $d_i+1$ } 
		\STATE {$temp := \{ \}$}
		\FOR{$k=1$ \TO $K$ } \label{a3}
		\IF{$n_k \neq 0 $}
		\STATE {$n_k := n_k - 1$} 
		\STATE {$temp := temp \cup \{c_k\}$} 
		\ENDIF
		\ENDFOR
		\STATE {Choose one of the member of $S$ randomly as $s$} \label{a4}
		\STATE {$v_s(t^+) := temp$} \label{b3}
		\IF{$|temp|==1$}
		\STATE $m_s(t^+):=temp$
		\ENDIF
		\STATE {$S := S \backslash \{s\}$} 
		\ENDFOR
		\ENDIF
	\end{algorithmic}
\label{alg:BDMV}
\end{algorithm}
%TODO : is the below text necessary or we should reference at the above text ?
%In line~\ref{a2} $n_k$ shows the number of agents which $c_k$  is in their value sets. 
In the $for$ loop in line~\ref{a3}, the algorithm tries to make a new value set from the choices which were in the previous value sets of agents involved in the interaction and gives it to one of these agents randomly if it has not received any new value set before (line~\ref{a4}).

In the following, we first prove the correctness of B-DMV algorithm. Then, we compare its performance with the pairwise version.

\begin{definition}
	A collection of sets is stable if for any $v_i(t)$ and $v_j(t)$ in the collection, we have: $ v_i(t)\subseteq v_j(t)$  if $|v_i(t)| \leq |v_j(t)|$.
	% We also denote the set of configuration vectors that are stable by $\mathcal{C}_0$. % \green{We have defined $C_0$ for initial configuration too!}
\end{definition}

\newtheorem{mylemma1}{Lemma}
\begin{mylemma1}
	After broadcasting interaction at agent $i$, the collection of value sets $N(i) \cup \{i\}$ is stable. 
	\label{lemma1}
\end{mylemma1}

\newtheorem{myproof1}{Proof}
\begin{proof}
	According to the description of assigning new value sets to agents in a broadcasting interaction, (lines~\ref{a2} to~\ref{b3} of B-DMV algorithm), each new value set is subset of new value sets with larger size. Therefore, this collection of new value sets is stable.
\end{proof}

In the following, we show that from any initial configuration, the algorithm eventually converges to a stable configuration after finite number of interactions. To do so, we define the following Lyapunov function on configuration vector:
\begin{equation}
V(C)= nK^2-\sum_{i=1}^n |v_i|^2,
\end{equation}
where $C=[v_1,\cdots,v_n]$  and show that the expected change in Lyapunov function after each interaction is negative.

\begin{theorem}
	B-DMV algorithm eventually outputs the correct result in finite number of interactions.
	\label{th:main}
\end{theorem}

\begin{proof}
	Based on Lemma \ref{lemma1}, assigning new value sets of agents in a broadcasting interaction results in a stable collection of new value sets for the agents involved in the interaction. Suppose that these agents execute DMVR algorithm instead of performing one broadcasting interaction. According to Lemma 4 in \cite{salehkaleybar2015distributed}, the value sets of these agents will reach to a stable collection after finite interactions, let say $r$ interactions. Assume that the initial configuration is $C_0$ and the execution path is: $C_0 \rightarrow C_1 \rightarrow ... \rightarrow C_r $.  In broadcasting version, after performing one interaction, we have: $ C_0 \rightarrow C'$ where $C'$ is a stable configuration. It can be shown that $C'$ is a permutation of $C_r$. Thus, we have: $V(C_r)=V(C')$. Moreover, according to Lemma 2 in \cite{salehkaleybar2015distributed}, in the pairwise interaction, we have: ($V(C_i) - V(C_{i-1}) \leq -\epsilon$) for any non-stable configuration $C_{i-1}$ where $\epsilon$ is a positive constant. Hence, for the broadcasting version, we can imply that:
	\begin{equation}
	\begin{split}
	V&(C_r) - V(C_0) = (V(C_r) - V(C_{r-1})) + \\&(V(C_{r-1}) - V(C_{r-2})) + ...
	+ (V(C_1) - V(C_0))\leq -r\epsilon. 
	\end{split}
	\end{equation}
	Thus, the expected change in the Lyapunov function after one broadcasting interaction is upper bounded by a negative constant if $C_0$ is not a stable configuration. Therefore, based on Foster's criteria (see \cite{asmussen2003applied}, page 21), B-DMV algorithm eventually converges to a stable configuration after finite number of interactions. In any stable configuration, there exist value sets of size one and all of them contain the choice in majority. After reaching a stable configuration, these value sets only update  memories of agents and all the memories will be eventually set to the majority choice.
\end{proof}

Like the pairwise version, B-DMV algorithm has two phases. In the first phase, the algorithm tries to reduce numbers of value sets of size one and converges to a stable configuration. In the second phase, the choice in majority will be spread through the network by updating memories of agents. By comparison of the process of the first phase in these two algorithms, one can argue that in broadcasting version, value sets of $d+1$ agents are updated  in each interaction. However, in the pairwise interaction, at each time, only value sets of two agents in the interaction are updated. So we expect that the runtime of this phase in broadcasting version would be smaller than the one of pairwise version. But in the second phase of broadcasting, since at each interaction, each value set of size one go randomly to one of the agents in interaction, so in the best-case scenario, the memory of $\frac{d}{2}$ agents in the interaction which does not have the correct value will be updated to the correct one. Intuitively, we expect that the runtime of the second phase in broadcasting does not improve a lot in comparison with the one of pairwise version. Moreover, in each interaction of broadcasting version, the number of messages which are transferred is $d+2$ (see the explanations in the beginning of Section \ref{sec:exp}) while it is just two messages in pairwise version. Hence, message complexity of B-DMV in the second phase can even increase with respect to the pairwise version. 
%%Since in each interaction in broadcasting the number of messages which are transferred is $d+2$ and it is just $2$ in pairwise, it can even cause the number of messages for broadcasting to get close to the pairwise.
%it can even increases the number of messages for broadcasting against the pairwise very much.
%TODO : the upper paragraph needs to be updated

%TODO: Maybe it is better to put curves of pairwise and boradcast in the same plot!
%TODO: Message and Time in the same page!

We evaluated these two algorithms experimentally on mesh network topology. Figure \ref{fig1} depicts the time and message complexities of both algorithms for solving the binary majority voting problem in a network with $100$ number of agents against percentage of agents selecting the choice $c_2$. Since the average number of interactions per unit of time is $n$, we define runtime (the measure of time complexity) as the number of interactions divided by $n$. 
%At the below figure, the results for time complexity have been shown:
%TODO : The above paragraph(51 of them): I've corrected it
\begin{figure*}
	\centering
	\begin{subfigure}[b]{0.32\linewidth}
		\includegraphics[width=\linewidth]{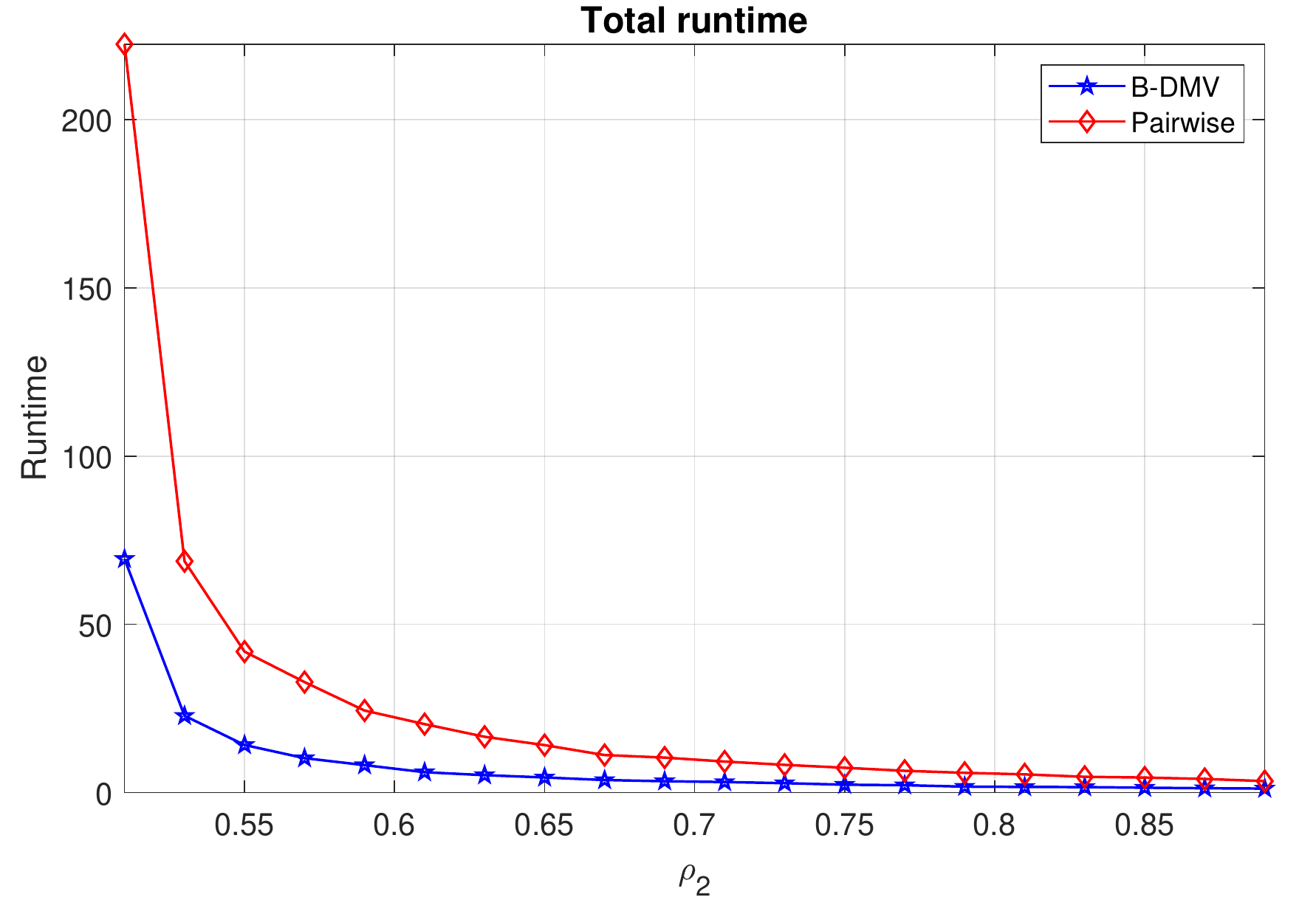}
		%%		\caption{Coffee.}
	\end{subfigure}
	\begin{subfigure}[b]{0.32\linewidth}
		\includegraphics[width=\linewidth]{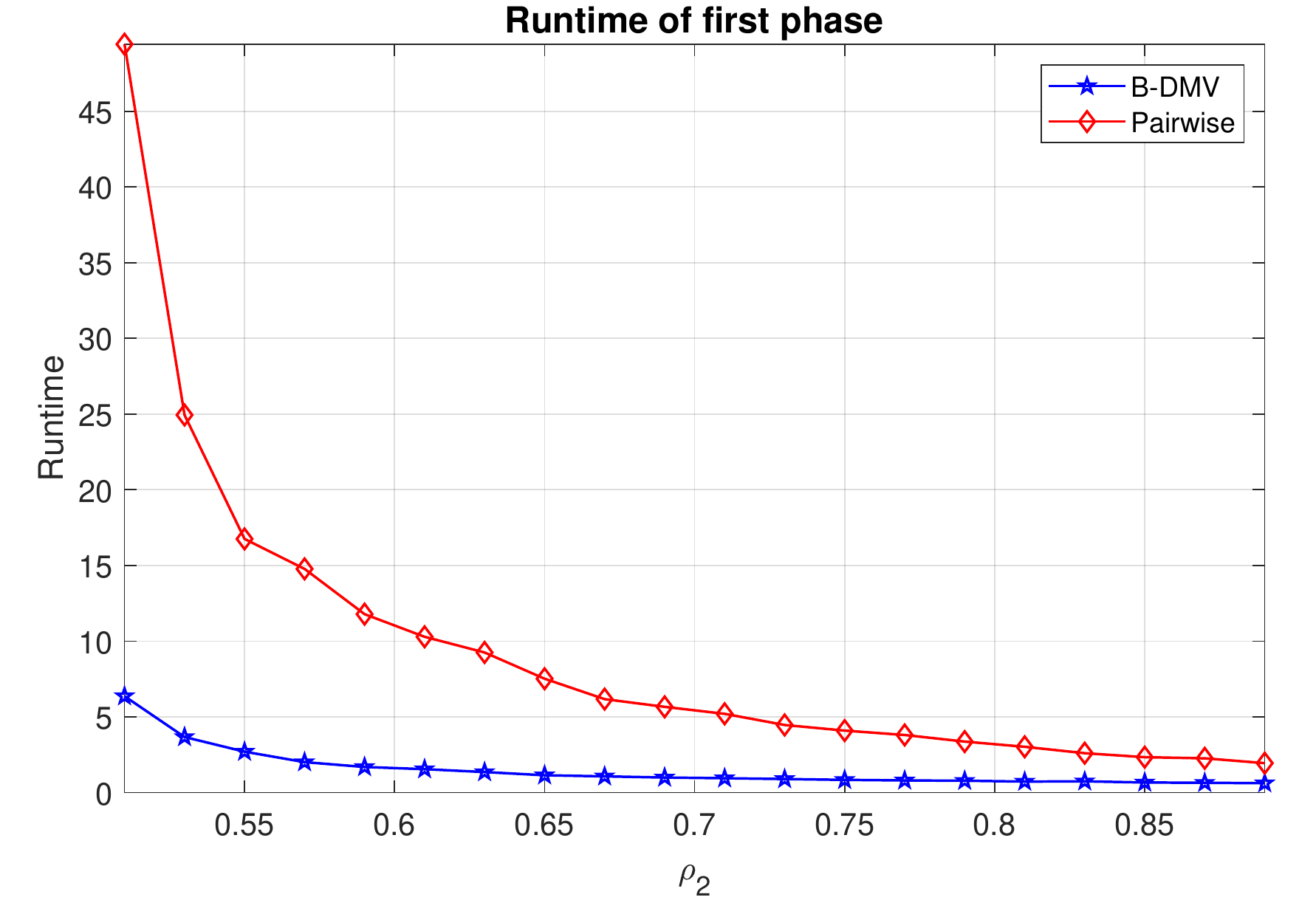}
	\end{subfigure}
	\begin{subfigure}[b]{0.32\linewidth}
		\includegraphics[width=\linewidth]{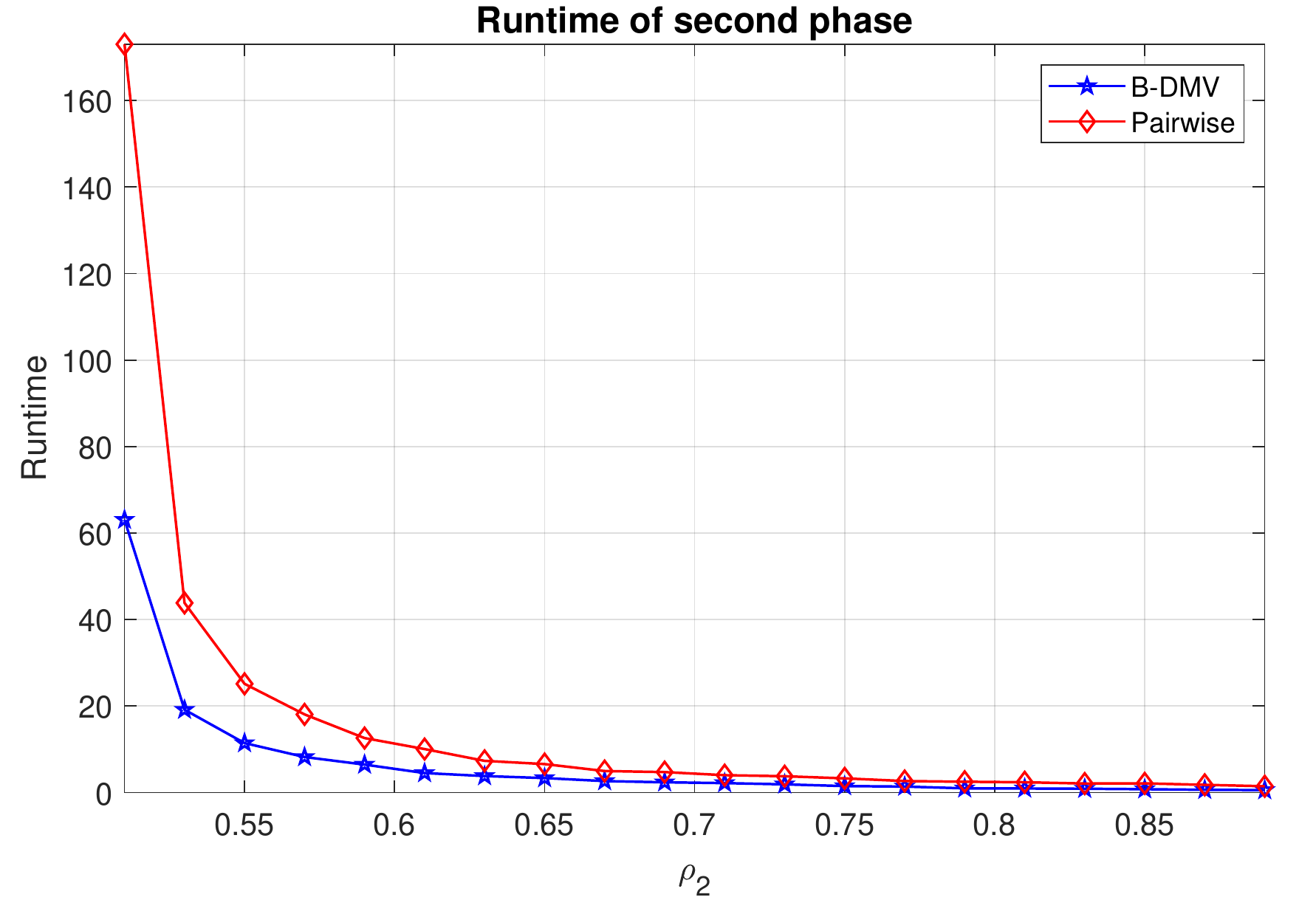}
	\end{subfigure}
	\begin{subfigure}[b]{0.32\linewidth}
		\includegraphics[width=\linewidth]{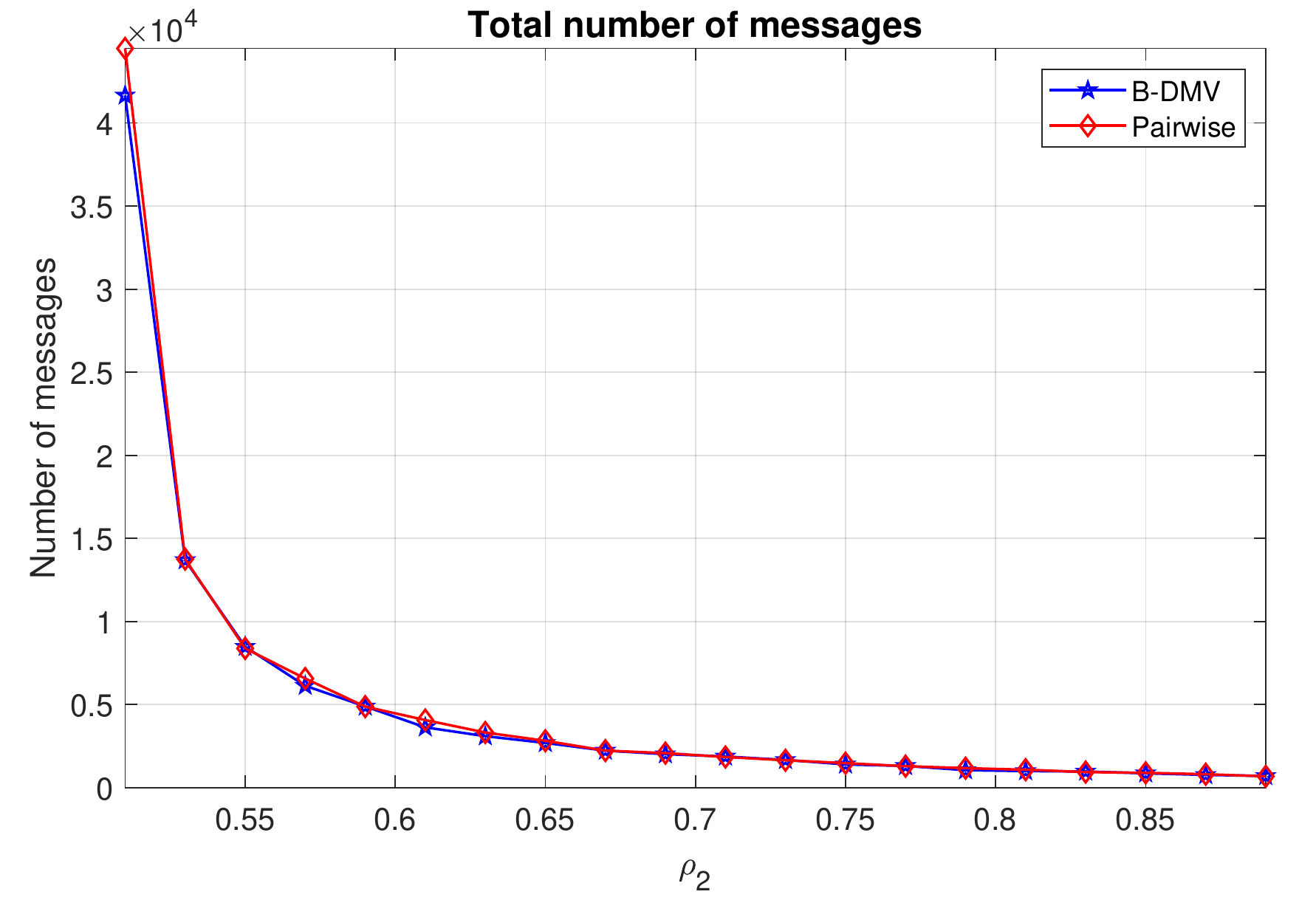}
	\end{subfigure}
	\begin{subfigure}[b]{0.32\linewidth}
		\includegraphics[width=\linewidth]{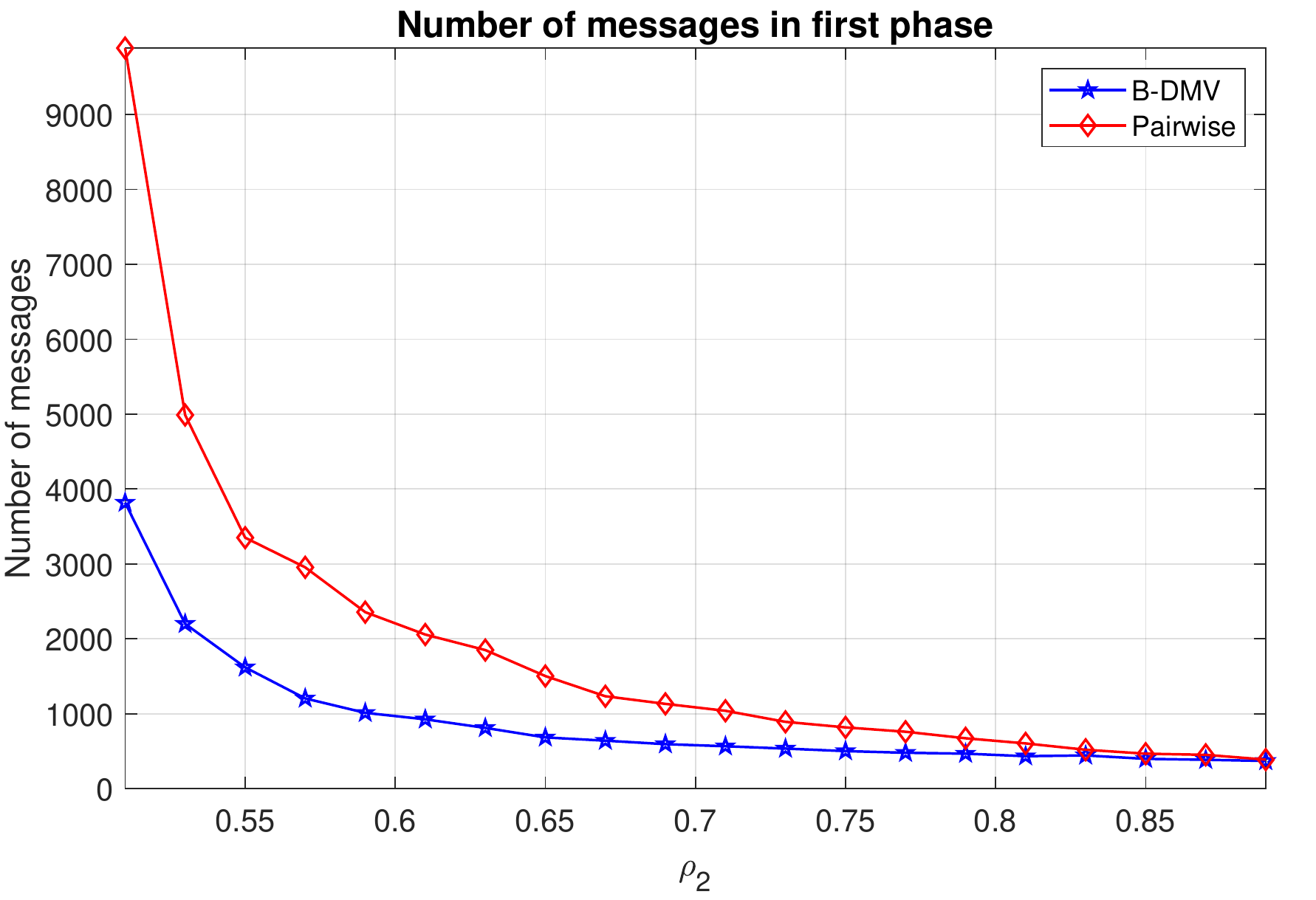}
	\end{subfigure}
	\begin{subfigure}[b]{0.32\linewidth}
		\includegraphics[width=\linewidth]{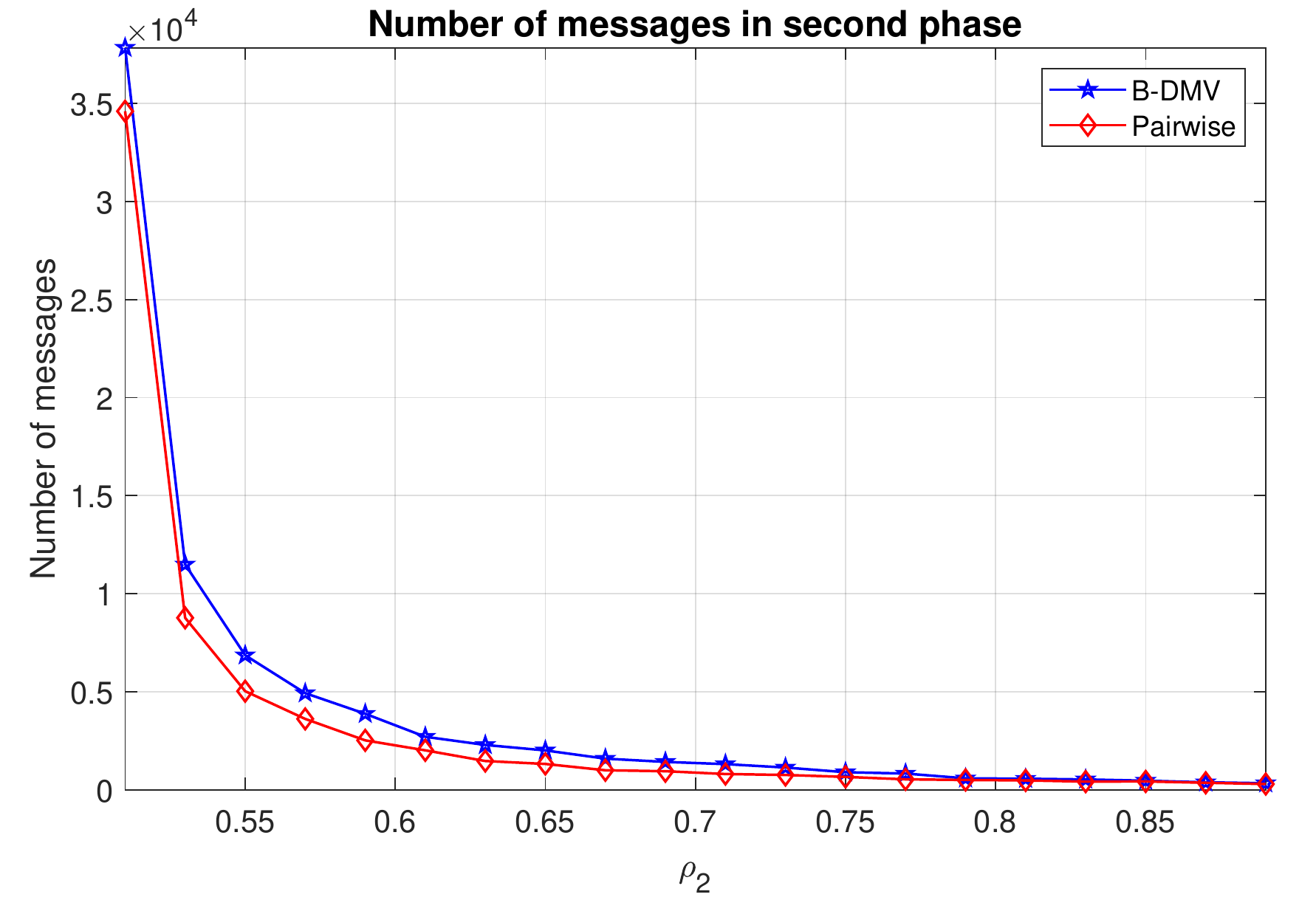}
	\end{subfigure}
	\caption{Time and message complexities of B-DMV algorithm and pairwise version (DMVR algorithm) in a mesh network with $n=100$ for $\rho_2$ in the range $0.51$ to $0.89$.}
	\label{fig1}
\end{figure*}
As can be seen, total time complexity decreases with respect to pairwise version (about $3$ times). Moreover, in the first phase, it decreases by a factor of $7$  while at the second one it decreases by a factor of $3$. As we expected the reduction of runtime in the first phase is greater than the second phase with respect to pairwise DMVR.
%TODO: change operations to interactions!
%The message complexities of the algorithms are shown in Figure \ref{fig1}.
%TODO : The below paragraphs need correction because of changing in message complexity results (I've editted.)
Furthermore, the message complexity of broadcast version is close to the one of pairwise version. In fact, as we expected in broadcast version, the message complexity of the first phase decreases. However, in the second phase of broadcast version, the message complexity is greater than the one in pairwise version.
%as we expected, the message complexity of the first phase decreases. However, in the second phase of broadcast version the message complexity is greater than the one in pairwise version.
%whereas in broadcast version the message complexity of the first phase decreases, in the second phase of it the message complexity is greater than the one in pairwise version.
%in broadcasting the number of messages transmitted by agents in an interaction is more than the one of pairwise version. Consequently, message complexity in broadcast version in the second phase is greater than the one in pairwise version. In fact, the import of increasing in number of messages in an interaction overcomes the effect of reduction in time complexity.

We can conclude that the broadcast version may not be better in message complexity (especially in the second phase). In the following, two improved version of B-DMV algorithm (called accelerated versions) will be proposed in order to improve the time and message complexities.

\subsection{Accelerated B-DMV 1}

As it has been mentioned, after the first phase, few agents have value sets of size one containing the correct result and their memories are also set to it. If they take part in an interaction, their single value sets of size one will go randomly to one of the agents in the interaction and the memory of that agent also be updated to the correct result. 
If the number of agents supporting different choices are close to each other, then only a few value sets of size one remain in the network after the first phase. These value sets need to walk randomly in the network to change all the memories to the correct results which might take too much time to finish the second phase. Here we resolve this problem by changing the updating rules in memories. Whenever an agent $i$ wakes up, it picks a choice that has the most repetitions in value sets of neighbor agents and set memories of them to that choice. Otherwise, updating of memories is like B-DMV algorithm. In particular, in the second phase, instead of passing value sets of size one, we can update all memories of neighbor agents. Please note that new updating rule of memories is applied from beginning of first phase. Hence, in some interactions memories of agents might be set to choices that are not in majority. However, at the end of first phase, we know that value sets of size one only contain the correct result. Hence, in the second phase, the value sets will update any memory with wrong choice to the correct result. 
%Here is the pseudo code of this new algorithm:

\begin{algorithm}[t]
	\caption{Accelerated B-DMV 1}
	\begin{algorithmic}[1]
		\STATE Initialization : $v_i(0)=$ Initial vote of agent $i$, $m_i(0)=$ Initial vote of agent $i$\;
		\IF{Agent $i$'s clock ticks at time $t$}
		\STATE Agent $i$ contacts with all of its neighbors, i.e., $N(i)$.
		\STATE Consider $v_i$ as $v_{i_0}$ and the value sets of $N(i)$ as $v_{i_1},...,v_{i_d}$ 
		\STATE $n_k := | \{j|c_k \in v_j, j=i_0,...,i_d\} |$
		\STATE $Cmax := \{ \}$\label{b5}
		\STATE $Cmax := \{c_l | n_l > n_j, j=1,...,k, j\neq i \} $\label{a5}
		\STATE {$S := N(i) \cup \{i\}$}
		\FOR {$u=1$ \TO $d_i+1$ } 
		\STATE {$temp := \{ \}$}
		\FOR{$k=1$ \TO $K$ } 
		\IF{$n_k \neq 0 $}
		\STATE {$n_k := n_k - 1$}
		\STATE {$temp := temp \cup \{c_k\}$} 
		\ENDIF
		\ENDFOR
		\STATE {Choose one of the member of $S$ randomly as $s$}
		\STATE {$v_s(t^+) := temp$} 
		\IF{$Cmax \neq \{ \}$}\label{a6}
		\STATE {$m_s(t^+):= Cmax $}
		\ELSIF {$|temp|==1$}
		\STATE $m_s(t^+):=temp$
		\ENDIF\label{a7}
		\STATE {$S := S \backslash \{s\}$} 
		\ENDFOR
		\ENDIF
	\end{algorithmic}
	\label{alg:bdmv1}
\end{algorithm}

The description of Accelerated B-DMV 1 is given in Algorithm \ref{alg:bdmv1}. As can be seen, this version is similar to the B-DMV and only the updating rule of memories is changed. In lines~\ref{b5} -~\ref{a5}, we construct a set $Cmax$ and add member if there is a choice like $c_l$ in value sets of agents of an interaction that have the most repetition. In lines~\ref{a6}-~\ref{a7}, the algorithm updates memories of all agents involved in the interaction.

\subsection{Accelerated B-DMV 2}

Similar to Accelerated B-DMV 1, the proposed algorithm in this part only changes the updating rule of memories in B-DMV. In the case that no choice in value set of agents in an interaction is repeated strictly more than the other choices, the algorithm checks the memories to see whether there is a choice with highest repetition. If so, the memory of all agents in the interaction will be updated to that choice. 

%TODO: please change c0-> c1 and c1-> c2 : done
In the next section, we will see that the Accelerated B-DMV 2 outperforms the first version in terms of message complexity and it is also better in time complexity. In the remainder of this part, we will show that Accelarted B-DMV 2 outputs the correct result for binary majority voting problem under some mild assumptions.
Assume that there are only two choices $c_1$ and $c_2$ and $c_2$ is in majority. At any time, we can partition agents into four groups:
$i)$ The group $V_1$ : the sets of agents whose value sets are $ \{ c_1 \}$. 
$ii)$ The group $V_2$ : the sets of agents whose value sets are $ \{ c_2 \}$. 
$iii)$ The group $M_1$ : the sets of agents whose value sets are not $\{c_1\}$ or $\{c_2\}$ and their memories are equal to $ \{ c_1 \}$. 
$iv)$ The group $M_2$ : the sets of agents whose value sets are not $\{c_1\}$ or $\{c_2\}$ and their memories are equal to $ \{ c_2 \}$. 

Let $v_1$, $v_2$, $m_1$, and $m_2$ be the size of sets $V_1$, $V_2$, $M_1$, and $M_2$, respectively.
Based on Theorem \ref{th:main}, the first phase eventually finishes in a finite number of interactions since updating rules of value sets are not changed in Accelerated B-DMV 2. Thus, it is just need to show that the number of agents having memory of $\{ c_1 \}$ converges to zero in the second phase. In other words, the output of all agents would be the choice in majority. To do so, for any configuration $C$, we define a Lyapanov function $L(C)=m_1$. We will show that the function $L(C)$ converges to zero after finite interactions. According to Foster's Criteria \cite{asmussen2003applied}, we need to prove that the expectation of $L(C^+)-L(C)$ is upper bounded by a negative constant where $C^+$ is any configuration that can be reached by $C$ through a broadcasting interaction. For simplicity of analysis, we consider mesh networks. Moreover, we assume that the following assumption holds true in our analysis.

\begin{assumption}
	At any time, we assume that the four groups $V_1$, $V_2$, $M_1$, and $M_2$ are uniformly distributed in the network and $m_2 > m_1$ in the second phase.
	\label{assum:main}
\end{assumption}

Experiments showed that the both assumptions mentioned above are fairly valid during the execution of the proposed algorithm.

\begin{theorem}
	Under Assumption \ref{assum:main}, in Accelerated B-DMV 2, the number of agents having memories of $\{ c_1 \}$ converges to zero after finite number of interactions.
\end{theorem}

\begin{proof}
	We know that $v_1 = 0$ in the second phase. In mesh network, each agent has four neighbors, i.e., $d=4$. Suppose that agent $i$ wakes up at time $t$. We denote the number of agents in $N(i) \cup \{ i \}$ belonging to groups $V_2$, $M_1$, and $M_2$ as $v_2^i$, $m_1^i$, and $m_2^i$, respectively. Hence, we know that : $v_2^i + m_1^i + m_2^i = d+1$ and total number of different assignments to $v_2^i$, $m_1^i$, and $m_2^i$ is equal to $ \left( \begin{array}{cc} d + 1 + 3 - 1 \\ 3 - 1 \end{array} \right) = 21 $. Furthermore, the probability of observing assignment $(v_2^i , m_1^i , m_2^i)$ is :
	\begin{align}
	P(v_2^i , m_1^i , m_2^i) = \frac{(d+1)!}{v_2^i! \times m_1^i! \times m_2^i!} \times p_{V_2}^{v_2^i} \times p_{M_1}^{m_1^i} \times p_{M_2}^{m_2^i},
	\end{align}
	where $p_{V_2} = \frac{v_2}{n}$, $p_{M_1} = \frac{m_1}{n}$, and $p_{M_2} = \frac{m_2}{n}$ according to Assumption \ref{assum:main}.
	Based on the above arguments, we can obtain the expectation of $L(C^+) - L(C)$ as follows:
	\begin{align}
	\sum\limits_{\substack{21 possible\\assignments\\of(v_2^i , m_1^i , m_2^i)}} P(v_2^i , m_1^i , m_2^i) \times \Delta(v_2^i , m_1^i , m_2^i),
	\end{align}
	where $\Delta(v_2^i , m_1^i , m_2^i)$ shows the change in the value of $m_1$ after executing interaction among $N(i) \cup \{i\}$ for the assignment $(v_2^i , m_1^i , m_2^i)$. We can simplify the above equation by setting the values of $\Delta(v_2^i , m_1^i , m_2^i)$ for each possible assignment of $(v_2^i , m_1^i , m_2^i)$ and obtain:
	%TODO: these parameters do not have superscript of j, right? : yes
	\begin{align}
	-c(m_1(m_2({m_2}^3- {m_1}^3) + 4{m_1}^3 v_2 + 4m_1 {m_2}^2(m_2 - m_1) +  \nonumber \\ 12 {m_1}^2 m_2 v_2 + 6 {m_1}^2 {v_2}^2  + 12 m_1 {m_2}^2 v_2 +  12 m_1 m_2 {v_2}^2 +  \nonumber \\ 4 m_1 {v_2}^3 + 4{m_2}^3 v_2 + 6{m_2}^2 {v_2}^2 + 4m_2 {v_2}^3 + {v_2}^4)),
	\end{align}
	where $c$ is a positive constant.
	Based on Assumption \ref{assum:main} ($m_2>m_1$), it can be shown that the above equation is always non-positive and this completes the proof.
\end{proof}

\section{Time and Message Complexities}
\label{sec:time}
We analyze the first and second phases of Accelerated B-DMV 1 by mean-field approximation \cite{pastor2015epidemic} and compare it with our simulation. In order to simplify analysis, we assume that there are only two choices $c_1$ and $c_2$ and each agent has exactly $d$ neighbors. Furthermore, we assume that the value set $\{c_1\}$ and $\{c_2\}$ are uniformly spread in the network at any time $t$. Let $x(t)$ and $y(t)$ be the number of agents whose value sets are $\{c_1\}$ and $\{c_2\}$ at time $t$, respectively. Moreover, assume that $c_2$ is in majority.

Recall that $v_1^i$ and $v_2^i$ are the number of agents in $N(i) \cup \{i\}$ having value sets $\{c_1\}$ and $\{c_2\}$, respectively. The probability of event that $v_1^i=0$ or $v_2^i=0$ is given as follows:
\begin{equation}
\begin{split}
\Pr(v_1^i=0 \lor v_2^i=0)&=\big(\Pr(v_1^i=0)+\Pr(v_2^i=0)\\& \qquad -(\Pr(v_1^i=0 \land v_2^i=0)\big)\\
&=\Bigg(\Big(1-\frac{x(t)}{n}\Big)^{d+1}+\Big(1-\frac{y(t)}{n}\Big)^{d+1}\\& \qquad -\Big(1-\frac{x(t)+y(t)}{n}\Big)^{d+1}\Bigg).
\end{split}
\end{equation}

The average reduction in number of agents having value set $\{c_1\}$ would be: $\mathbb{E}[\min(v_1^i,v_2^i)]$ which is greater than $1-\Pr(v_1^i=0 \lor v_2^i=0)$.
%\begin{align}
%E[min(v_0^i,v_1^i)] & =  E[min(v_0^i,v_1^i)|v_0^i<v_1^i]\times P(v_0^i<v_1^i)\nonumber \\
%  & + E[min(v_0^i,v_1^i)|v_0^i>v_1^i]\times P(v_0^i>v_1^i) \nonumber \\
%  &= \frac{x(t)}{n}\times(d+1)\times\frac{y(t)}{x(t)+y(t)} \nonumber \\
%  & + \frac{y(t)}{n}\times(d+1)\times\frac{x(t)}{x(t)+y(t)} \nonumber \\
%  &= \frac{x(t)y(t)}{x(t)+y(t)}\times\frac{d+1}{n}\times2
%\label{eq:1}
%\end{align}
%%In other words, we have:\\
%%\begin{align}
%%\frac{x(t)(x(t)+y(0)-x(0))}{2x(t)+y(0)-x(0)}\times\frac{d+1}{n}\times2
%%\label{eq:1}
%%\end{align}
%where $E[v_0^i] = \frac{x(t)}{n}(d+1)$, $E[v_1] = \frac{y(t)}{n}(d+1)$, $P(c_0<c_1)=\frac{y(t)}{x(t)+y(t)}$ due to uniform spread of value sets assumption. 
Now, by mean-field approximation, we set the rate of reduction of $x(t)$ equal to the average reduction, i.e., $dx(t)/dt=-\mathbb{E}[\min(v_1^i,v_2^i)]$. Since $\mathbb{E}[\min(v_1^i,v_2^i)]$ is lower bounded by $1-\Pr(v_1^i=0 \lor v_2^i=0)$, we have:
\begin{align}
\label{eq:anal_first}
\begin{split}
\frac{dx(t)}{dt} \leq  -\Bigg(1- \Bigg(\Big(1-\frac{x(t)}{n}\Big)^{d+1}+\Big(1-\frac{y(t)}{n}\Big)^{d+1}-\\ \Big(1-\frac{x(t)+y(t)}{n}\Big)^{d+1}\Bigg)\Bigg).
\end{split}
\end{align}

%TODO: change all operations in figures to interactions

%TODO: Maybe more figs for different rhos!
The result of simulation and the differential equation based on above inequality against number of interactions for a mesh network with $100$ agents and three different $\rho_2$s ($\rho_2=0.7$, $\rho_2=0.8$, and $\rho_2=0.9$) are given in Figure \ref{fig3}. The simulation result is averaged over $100$ different runs of Accelerated B-DMV 1 algorithm. By comparison of the two curves, it can be seen that our analysis is fairly a good upper bound on the runtimes of the algorithm. 
\begin{figure}[t]
	\begin{center}
		\includegraphics[width=9cm]{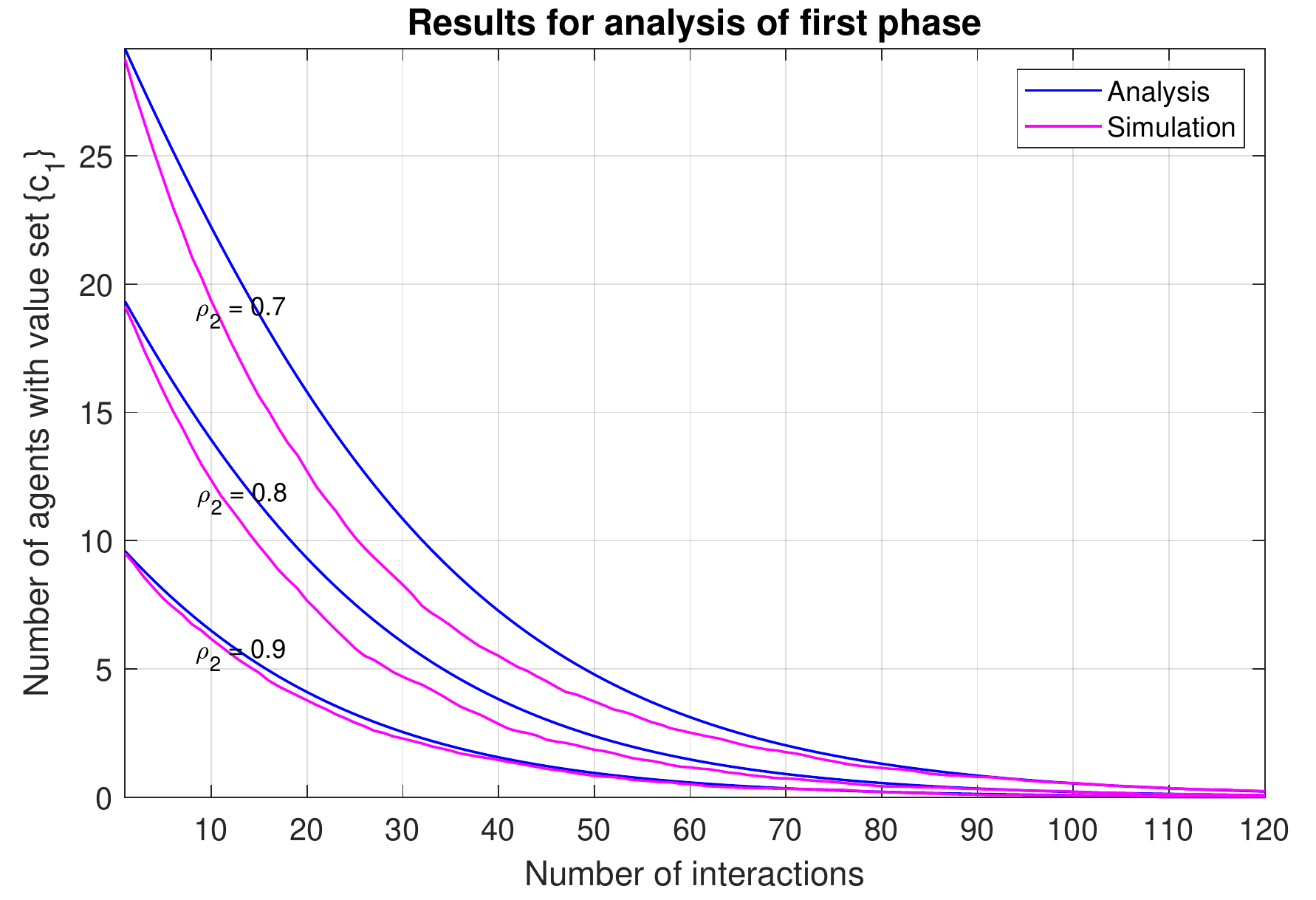}
		\caption{Results for analysis of the first phase.}
		\label{fig3}
	\end{center}
\end{figure}

We approximate the right hand side of \eqref{eq:anal_first} by keeping only first three terms of binomial expansion:

\begin{align}
\begin{split}
-\Bigg(&1-
\Bigg(\Big(1-\frac{x(t)}{n}\Big)^{d+1}+\Big(1-\frac{y(t)}{n}\Big)^{d+1}- \\ &\qquad\qquad\Big(1-\frac{x(t)+y(t)}{n}\Big)^{d+1}\Bigg)\Bigg) \\ &\qquad\approx -\frac{d(d+1)}{n^2} \times x(t) \times y(t) \\ &\qquad= -\frac{d(d+1)}{n^2} \times x(t) \times (x(t)+y(0)-x(0)),
\end{split}
\end{align}
where the last equality is due to the fact that that the difference between $y(t)$ and $x(t)$ is always a constant $y(0)-x(0)$. Thus, we can write:
\begin{align}
\begin{split}
 \frac{dx(t)}{dt} &\approx -\frac{d(d+1)}{n^2} \times x(t) \times (x(t)+y(0)-x(0)), 
 \\ \	\Rightarrow x(t) & = \frac{y(0)-x(0)}{\frac{y(0)}{x(0)}\times e^{\frac{d(d+1)}{n^2}(y(0)-x(0))t} - 1}.
\end{split}
\end{align}

From the above equation, we obtain an upper bound on time $t_{c_1}$ which is the time that $x(t)$ becomes less than or equal to $1$:

\begin{align}
\begin{split} 
t_{c_1} \leq \frac{n}{d(d+1)(\rho_2-\rho_1)}\Bigg({\log\Big(\frac{\rho_1}{\rho_2}\Big) + \log\Big(n(\rho_2-\rho_1)+1\Big)}\Bigg).
\end{split}
\label{eq:tc1}
\end{align}
%where $\rho_2$ and $\rho_x$ are $\frac{y(0)}{n}$ and $\frac{x(0)}{n}$, respectively.

 %As you can see, in this algorithm $t_{c_1}$ is in order of $O(n \times log(\rho_y - \rho_x))$.

%\green{Shouldn't we write $t_{c_1}\geq ...$}
Now we study  the second phase. We know that there is no value set of size one from minority choices $c_1$ at the beginning of the second phase. Here we show number of agents whose memories are not $\{ c_2 \}$ by $z(t)$. At the beginning of second phase, we assume that $z(t)$ is less than $x(0)$. By using mean field approximation, we have the following differential equation:
%TODO: we should check this eq again!
\begin{align}
\frac{dz(t)}{dt}  \leq  - \frac{z(t)}{n} \Bigg(1-\Big(1-\frac{y(0)-x(0)}{n}\Big)^d\Bigg),
\label{eq:1}
\end{align}  
where the absolute value of right hand side of above inequality is the probability that the agent initiating the interaction has the memory of $\{c_1\}$ and there is at least one agent in its neighbor set having the value set of $\{c_2\}$.
%Again like the previous part, we can obtain a lower bound for runtime by assuming an upper bound for $z(t)$.

If we assume that $\frac{y(0)-x(0)}{n}$ is small enough, we can simplify the above equation in the following form:
\begin{align}
\frac{dz(t)}{dt}  \leq  - \frac{d}{n^2}\times z(t) (y(0) - x(0)),
\end{align} 

In Figure \ref{fig4}, the result of the same simulation for the second phase is given. As can be seen, the result of analysis is an upper bound for the one from simulation.

\begin{figure}[t]
	\begin{center}
		\includegraphics[width=9cm]{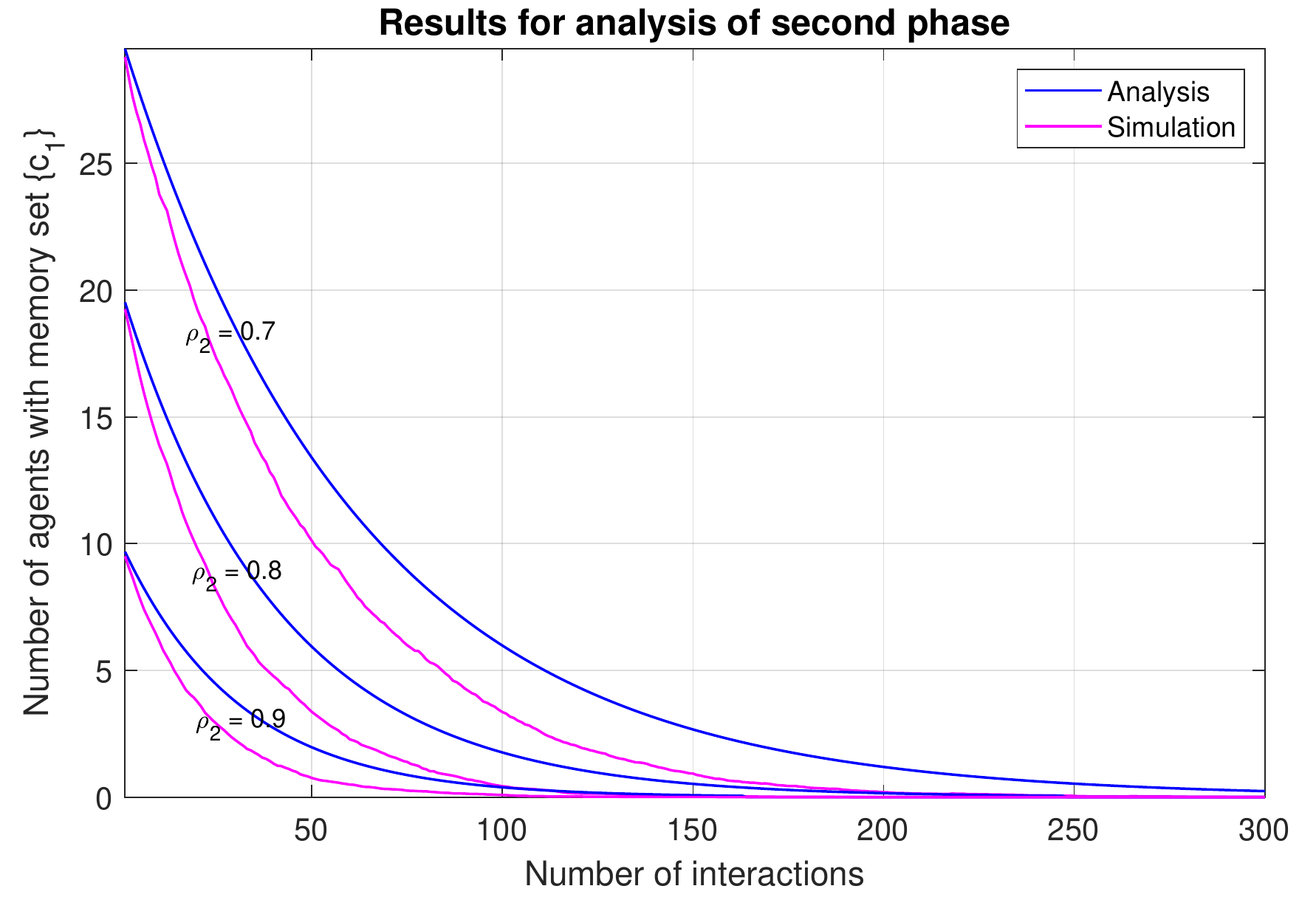}
		\caption{Results for analysis of the second phase.}
		\label{fig4}
	\end{center}
\end{figure}

From the above equation, we can obtain an upper bound on time $t_{c_2}$ which is the time that $z(t)$ becomes less than or equal to one:

\begin{align}
\begin{split}
t_{c_2} \leq \frac{n}{d \times (\rho_2 - \rho_1)} \log(n \rho_1).
\end{split}
\label{eq:tc2}
\end{align}

Combining \eqref{eq:tc1} and \eqref{eq:tc2}, the total time complexity is in the order of $O(n\log(n\rho_2)/(d(\rho_2-\rho_1)))$. Moreover, the number of transmitted messages is ($d+2$)-times of number of interactions (see the explanations at the beginning of Section \ref{sec:exp}). Hence, the message complexity would be in the order of $O(n\log(n\rho_2)/(\rho_2-\rho_1))$.

%Here $t_{c_2}$ is in order of $O(n\times log(n \rho_x))$. Totally, for both first and second phase $t_c = t_{c_1} + t_{c_2}$ which is the convergence time has been calculated at the below:
%\begin{align}
%\begin{split}
%t_c \geq \frac{n}{d(d+1)(\rho_y - \rho_x)} log\Big(\frac{\rho_x}{\rho_y} (\rho_y-\rho_x+1) (n \rho_x)^{d+1}\Big).     
%\end{split}
%\end{align}
%So, we can say that totally $t_c$ is in order of $O(n\times log(n\rho_x))$.

\section{Experiments}
\label{sec:exp}
In this section, we evaluate the performance of proposed algorithms in terms of time and message complexities.
%TODO: Why are you defining new parameters?
% We define two parameters in this section and then we analyze these algorithms by using these parameters. The first parameter is an \textbf{"Operation"}. Each interaction in pairwise or broadcasting is called an operation and in the following, for time complexity we use from "operation" parameter. The second parameter is a \textbf{"Message"}.
In pairwise population protocols, in each interaction, two messages are transferred. However, in broadcast version, in each interaction, there are three steps. First, the agent which initiates the interaction, broadcasts a message to all of its neighbors. Then, all the neighbors send their value sets to the initiator. Finally, the initiator broadcasts the new value sets of all neighbors. Thus, in broadcast population protocol, in each interaction, $d + 2$ number of messages are transmitted where $d$ is the degree of each agent. In our simulations, we report the number of transmitted messages based on this observation.
%\\
%Now we want to analyze time and message complexity (for time complexity we count number of operations and for message complexity we count number of messages). 

\begin{figure*}
	\centering
	\begin{subfigure}[b]{0.48\linewidth}
		\includegraphics[width=\linewidth]{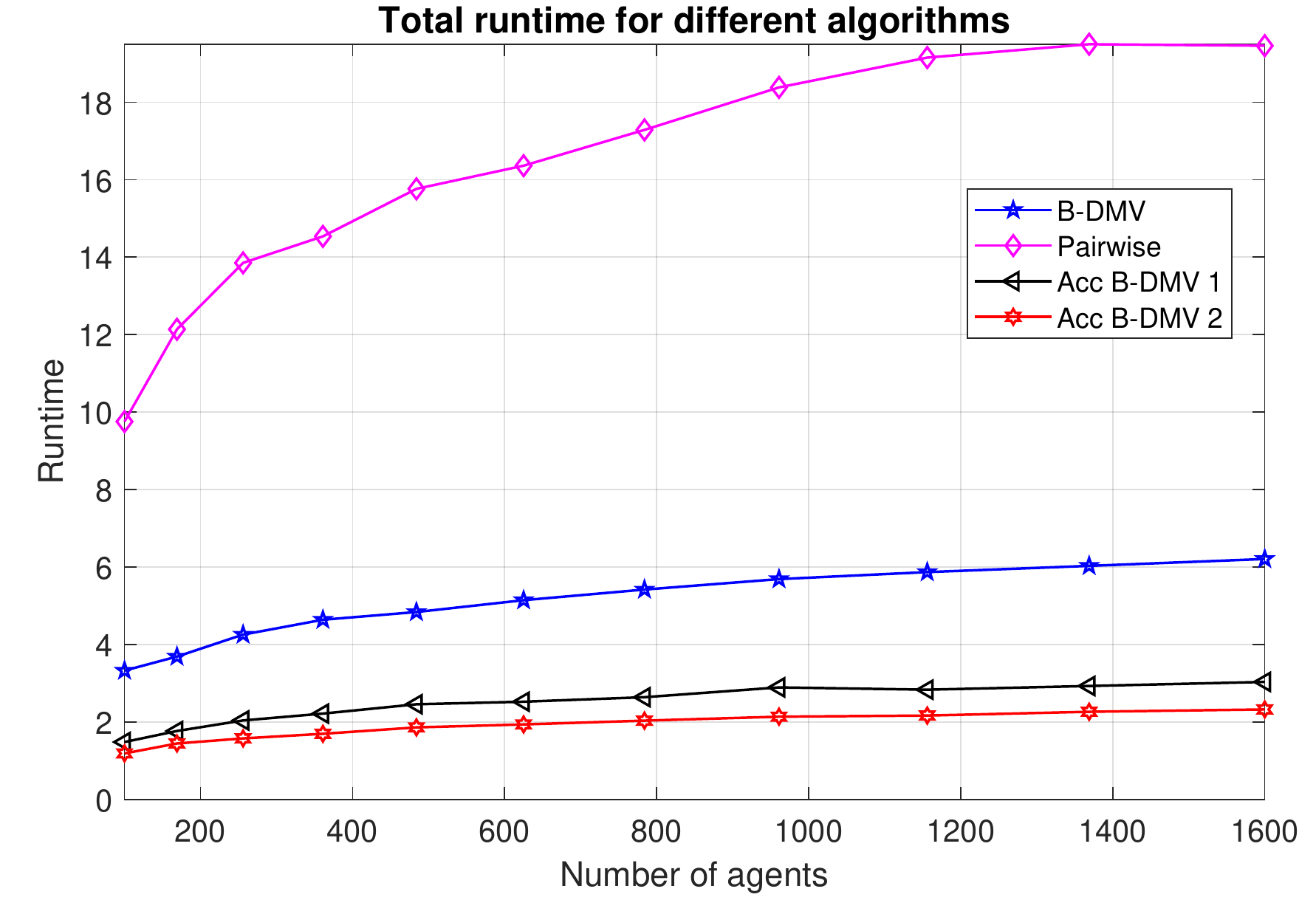}
		\caption{Runtime versus number of agents}
	\end{subfigure}
	\begin{subfigure}[b]{0.48\linewidth}
		\includegraphics[width=\linewidth]{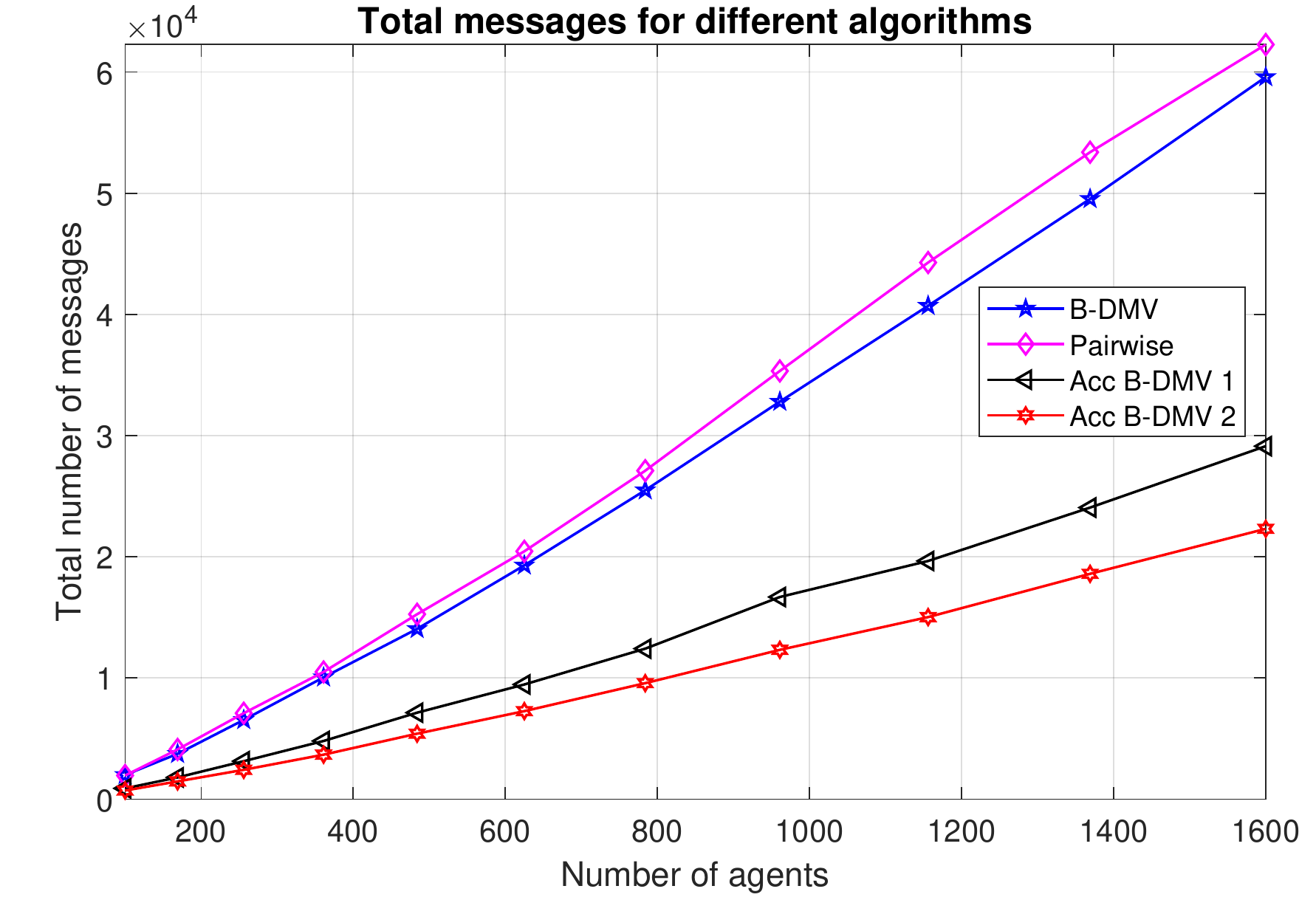}
		\caption{Total number of messages versus number of agents}
	\end{subfigure}
	\caption{Runtime and transmitted messages for different algorithms in a mesh network topology with $\rho_2 = 0.7$}
	\label{fig5}
\end{figure*}

%\begin{figure*}
%	\centering
%	\begin{subfigure}[b]{0.35\linewidth}
%		\includegraphics[width=\linewidth]{plot/fig6_1.pdf}
%		\caption{With Ring topology}
%	\end{subfigure}
%	\begin{subfigure}[b]{0.35\linewidth}
%		\includegraphics[width=\linewidth]{plot/fig6_2.pdf}
%		\caption{Without Ring topology}
%	\end{subfigure}
%	\caption{Number of interactions for Accelerated B-DMV 2 algorithm until convergence for different  network topologies.}
%	\label{fig6}
%\end{figure*}

%\begin{figure*}
%	\centering
%	\begin{subfigure}[b]{0.35\linewidth}
%		\includegraphics[width=\linewidth]{plot/fig6_1.pdf}
%		\caption{With Ring topology}
%	\end{subfigure}
%	\begin{subfigure}[b]{0.35\linewidth}
%		\includegraphics[width=\linewidth]{plot/fig6_2.pdf}
%		\caption{Without Ring topology}
%	\end{subfigure}
%	\caption{Number of interactions for Accelerated B-DMV 2 algorithm until convergence for different  network topologies.}
%	\label{fig6}
%\end{figure*}

\begin{figure}[h]
	\begin{center}
		\includegraphics[width=9.2cm]{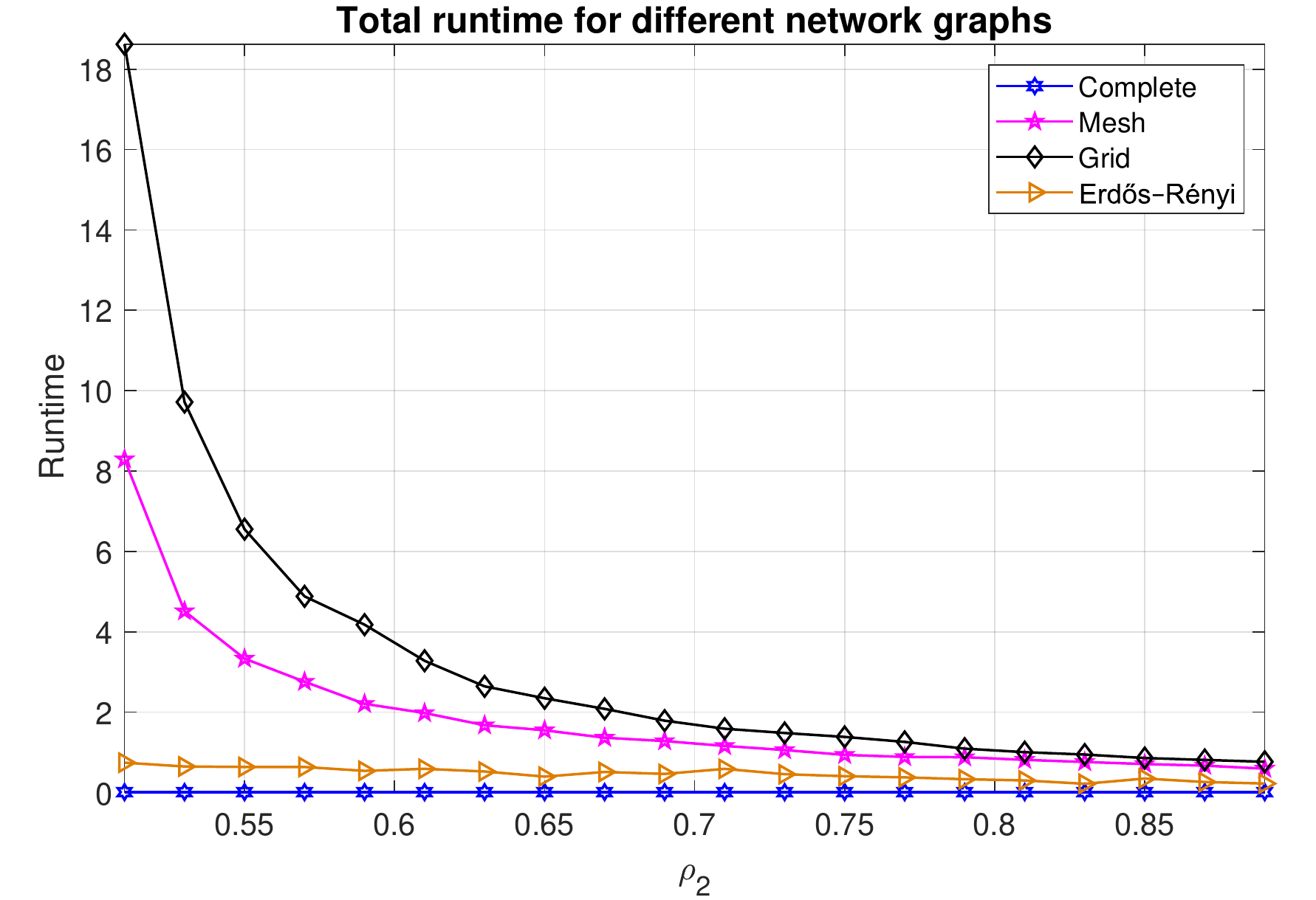}
		\caption{Runtime of Accelerated B-DMV 2 algorithm until convergence for different  network topologies.}
			\label{fig6}
	\end{center}
\end{figure}
We simulated four different algorithms (Pairwise, B-DMV, Accelerated B-DMV 1, and Accelerated B-DMV 2) on mesh networks with different number of agents (from $100$ to $1600$). The results are given in Figure \ref{fig5}. 
As we expected, for the time complexity, Accelerated B-DMV 2 has the best performance and Accelerated B-DMV 1 is better than B-DMV. Moreover, the time complexities of all algorithms scale with the diameter of network which is in the order of $\sqrt{n}$. Note that we need to wait for at least the diameter of network so that all nodes have the same memory.
%Total curve of these algorithms is like square root curve and it is due to the fact that parallel time complexity is calculated and it is related to the diagonal of our network which is a coefficient of square root of $n$ (Recall that $n$ shows number of agents).
%TODO: It needs to be updated
For the message complexity, Accelerated B-DMV 2 algorithm again has the best performance and Accelerated B-DMV 1 is better than B-DMV. Moreover, the message complexity of B-DMV is very close to pairwise one. 
%TODO: which algorithms? : Accelerated B-DMV 2 :

%TODO: check spell of Erdos Reyni : done

We also evaluated the performance of Accelerated B-DMV 2 algorithm on different topologies (Complete graphs, mesh graphs, grid, and Erd\H{o}s-R\'{e}nyi graphs) with $100$ number of agents. Here, we report time complexities of this algorithm on the mentioned topologies for different $\rho_2$ in Figure \ref{fig6}. As we expected, the algorithm converges faster in networks with smaller diameters (such as complete graphs or Erd\H{o}s-R\'{e}nyi graphs). Moreover, these results are consistent with the analysis in Section \ref{sec:time} where the time complexity is a factor of $1/(\rho_2-\rho_1)=1/(2\rho_2-1)$.

\section{Conclusion and Future Works}
In this paper, we introduced broadcasting population protocol model for the networks that each agent can send messages to a subset of agents simultaneously such as in wireless networks. We proposed two distributed algorithms for the problem of multi-choice majority voting in this model. We proved the correctness of these algorithms and analyzed their time and message complexities. Experimental results showed that for different network topologies, the proposed algorithms have better performance in both time and message complexities with respect to previous works proposed for pairwise population protocols. As a future work, one can study other problems such as leader election or consensus in this model. Moreover, it is interesting to check whether we can speed up the runtime of algorithms designed for pairwise population protocols by adopting them for broadcasting population protocols.

% or


%\appendices
%% you can choose not to have a title for an appendix
%% if you want by leaving the argument blank
%\section{}
%% use section* for acknowledgment
%\ifCLASSOPTIONcompsoc
%  % The Computer Society usually uses the plural form
%\else
%  % regular IEEE prefers the singular form
%\fi

\ifCLASSOPTIONcaptionsoff
  \newpage
\fi



%\bibliographystyle{IEEEtran}

\bibliographystyle{IEEEtran}
\bibliography{ref}

\begin{thebibliography}{10}
\providecommand{\url}[1]{#1}
\csname url@samestyle\endcsname
\providecommand{\newblock}{\relax}
\providecommand{\bibinfo}[2]{#2}
\providecommand{\BIBentrySTDinterwordspacing}{\spaceskip=0pt\relax}
\providecommand{\BIBentryALTinterwordstretchfactor}{4}
\providecommand{\BIBentryALTinterwordspacing}{\spaceskip=\fontdimen2\font plus
\BIBentryALTinterwordstretchfactor\fontdimen3\font minus
  \fontdimen4\font\relax}
\providecommand{\BIBforeignlanguage}[2]{{%
\expandafter\ifx\csname l@#1\endcsname\relax
\typeout{** WARNING: IEEEtran.bst: No hyphenation pattern has been}%
\typeout{** loaded for the language `#1'. Using the pattern for}%
\typeout{** the default language instead.}%
\else
\language=\csname l@#1\endcsname
\fi
#2}}
\providecommand{\BIBdecl}{\relax}
\BIBdecl

\bibitem{lynch1996distributed}
N.~A. Lynch, \emph{Distributed algorithms}.\hskip 1em plus 0.5em minus
  0.4em\relax Elsevier, 1996.

\bibitem{navlakha2015distributed}
S.~Navlakha and Z.~Bar-Joseph, ``Distributed information processing in
  biological and computational systems,'' \emph{Communications of the ACM},
  vol.~58, no.~1, pp. 94--102, 2015.

\bibitem{chen2014deterministic}
H.-L. Chen, D.~Doty, and D.~Soloveichik, ``Deterministic function computation
  with chemical reaction networks,'' \emph{Natural computing}, vol.~13, no.~4,
  pp. 517--534, 2014.

\bibitem{alistarh2015polylogarithmic}
D.~Alistarh and R.~Gelashvili, ``Polylogarithmic-time leader election in
  population protocols,'' in \emph{International Colloquium on Automata,
  Languages, and Programming}.\hskip 1em plus 0.5em minus 0.4em\relax Springer,
  2015, pp. 479--491.

\bibitem{doty2018stable}
D.~Doty and D.~Soloveichik, ``Stable leader election in population protocols
  requires linear time,'' \emph{Distributed Computing}, vol.~31, no.~4, pp.
  257--271, 2018.

\bibitem{angluin2005self}
D.~Angluin, J.~Aspnes, M.~J. Fischer, and H.~Jiang, ``Self-stabilizing
  population protocols,'' in \emph{International Conference On Principles Of
  Distributed Systems}.\hskip 1em plus 0.5em minus 0.4em\relax Springer, 2005,
  pp. 103--117.

\bibitem{berenbrink2018simple}
P.~Berenbrink, D.~Kaaser, P.~Kling, and L.~Otterbach, ``Simple and efficient
  leader election,'' in \emph{1st Symposium on Simplicity in Algorithms (SOSA
  2018)}.\hskip 1em plus 0.5em minus 0.4em\relax Schloss
  Dagstuhl-Leibniz-Zentrum fuer Informatik, 2018.

\bibitem{aspnes2016time}
J.~Aspnes, J.~Beauquier, J.~Burman, and D.~Sohier, ``Time and space optimal
  counting in population protocols,'' \emph{arXiv preprint arXiv:1611.07238},
  2016.

\bibitem{beauquier2015space}
J.~Beauquier, J.~Burman, S.~Clavi{\`e}re, and D.~Sohier, ``Space-optimal
  counting in population protocols,'' in \emph{International Symposium on
  Distributed Computing}.\hskip 1em plus 0.5em minus 0.4em\relax Springer,
  2015, pp. 631--646.

\bibitem{izumi2014space}
T.~Izumi, K.~Kinpara, T.~Izumi, and K.~Wada, ``Space-efficient self-stabilizing
  counting population protocols on mobile sensor networks,'' \emph{Theoretical
  Computer Science}, vol. 552, pp. 99--108, 2014.

\bibitem{beauquier2007self}
J.~Beauquier, J.~Clement, S.~Messika, L.~Rosaz, and B.~Rozoy,
  ``Self-stabilizing counting in mobile sensor networks,'' in \emph{Proceedings
  of the twenty-sixth annual ACM symposium on Principles of distributed
  computing}.\hskip 1em plus 0.5em minus 0.4em\relax ACM, 2007, pp. 396--397.

\bibitem{angluin2008simple}
D.~Angluin, J.~Aspnes, and D.~Eisenstat, ``A simple population protocol for
  fast robust approximate majority,'' \emph{Distributed Computing}, vol.~21,
  no.~2, pp. 87--102, 2008.

\bibitem{benezit2011distributed}
F.~B{\'e}n{\'e}zit, P.~Thiran, and M.~Vetterli, ``The distributed multiple
  voting problem,'' \emph{IEEE Journal of Selected Topics in Signal
  Processing}, vol.~5, no.~4, pp. 791--804, 2011.

\bibitem{salehkaleybar2015distributed}
S.~Salehkaleybar, A.~Sharif-Nassab, and S.~J. Golestani, ``Distributed
  voting/ranking with optimal number of states per node,'' \emph{IEEE
  Transactions on Signal and Information Processing over Networks}, vol.~1,
  no.~4, pp. 259--267, 2015.

\bibitem{gasieniec2017deterministic}
L.~Gasieniec, D.~Hamilton, R.~Martin, P.~G. Spirakis, and G.~Stachowiak,
  ``Deterministic population protocols for exact majority and plurality,''
  \emph{Leibniz International Proceedings in Informatics, LIPIcs}, vol.~70, pp.
  14--1, 2017.

\bibitem{alistarh2018space}
D.~Alistarh, J.~Aspnes, and R.~Gelashvili, ``Space-optimal majority in
  population protocols,'' in \emph{Proceedings of the Twenty-Ninth Annual
  ACM-SIAM Symposium on Discrete Algorithms}.\hskip 1em plus 0.5em minus
  0.4em\relax SIAM, 2018, pp. 2221--2239.

\bibitem{alistarh2015fast}
D.~Alistarh, R.~Gelashvili, and M.~Vojnovi{\'c}, ``Fast and exact majority in
  population protocols,'' in \emph{Proceedings of the 2015 ACM Symposium on
  Principles of Distributed Computing}.\hskip 1em plus 0.5em minus 0.4em\relax
  ACM, 2015, pp. 47--56.

\bibitem{benezit2009interval}
F.~B{\'e}n{\'e}zit, P.~Thiran, and M.~Vetterli, ``Interval consensus: from
  quantized gossip to voting,'' in \emph{2009 IEEE International Conference on
  Acoustics, Speech and Signal Processing}.\hskip 1em plus 0.5em minus
  0.4em\relax IEEE, 2009, pp. 3661--3664.

\bibitem{alistarh2017time}
D.~Alistarh, J.~Aspnes, D.~Eisenstat, R.~Gelashvili, and R.~L. Rivest,
  ``Time-space trade-offs in population protocols,'' in \emph{Proceedings of
  the Twenty-Eighth Annual ACM-SIAM Symposium on Discrete Algorithms}.\hskip
  1em plus 0.5em minus 0.4em\relax SIAM, 2017, pp. 2560--2579.

\bibitem{asmussen2003applied}
S.~Asmussen, ``Applied probability and queues, (2003),'' 2003.

\bibitem{pastor2015epidemic}
R.~Pastor-Satorras, C.~Castellano, P.~Van~Mieghem, and A.~Vespignani,
  ``Epidemic processes in complex networks,'' \emph{Reviews of modern physics},
  vol.~87, no.~3, p. 925, 2015.

\end{thebibliography}

%\begin{thebibliography}{1}

%\end{thebibliography}


%\end{IEEEbiography}

%%\begin{IEEEbiographynophoto}{John Doe}
%%Biography text here.
%%\end{IEEEbiographynophoto}

%\newpage



%\vfill


\end{document}